\documentclass[fleqn]{eptcs}
\providecommand{\event}{GandALF 2011}
\usepackage{breakurl}
\usepackage{amsmath,amssymb,amsthm,algorithm}

\usepackage{tikz}
\usetikzlibrary{decorations.pathmorphing}
\usetikzlibrary{shapes.geometric}
\usetikzlibrary{arrows}

\renewcommand{\phi}{\varphi}

\newcommand{\nats}{\mathbb{N}}
\renewcommand{\implies}{\rightarrow}
\newcommand{\dual}[1]{\overline{#1}}
\newcommand{\coloneq}{\mathop{:=}}

\newcommand{\cceq}{\mathop{::=}}

\newcommand{\F}{\mathbf{F}}
\newcommand{\G}{\mathbf{G}}
\newcommand{\U}{\mathbf{U}}
\newcommand{\X}{\mathbf{X}}
\newcommand{\R}{\mathbf{R}}
\newcommand{\true}{\mathbf{tt}}
\newcommand{\false}{\mathbf{ff}}
\newcommand{\Xvar}{\mathcal{X}}
\newcommand{\Yvar}{\mathcal{Y}}
\newcommand{\ltl}{\mathrm{LTL}}
\newcommand{\pltlf}{\pltl_{\F}}
\newcommand{\pltlg}{\pltl_{\G}}
\newcommand{\pltl}{\mathrm{PLTL}}
\newcommand{\var}{\mathrm{var}}
\newcommand{\prompt}{\mathrm{PROMPT}$\textendash$\ltl}

\newcommand{\aut}{\mathfrak{A}}

\newcommand{\cl}{\mathrm{cl}}
\newcommand{\autp}{\mathfrak{P}}

\newcommand{\acc}{\mathrm{Acc}}
\newcommand{\curlyF}{\mathcal{F}}
\newcommand{\curlyC}{\mathcal{C}}
\newcommand{\bigo}{\mathcal{O}}
\newcommand{\W}{\mathcal{W}}

\newcommand{\game}{\mathcal{G}}
\newcommand{\arena}{\mathcal{A}}

\newcommand{\mem}{\mathcal{M}}
\newcommand{\update}{\mathrm{upd}}
\newcommand{\nextmove}{\mathrm{nxt}}

\newcommand{\pspace}{\textbf{\textsc{Pspace}}}
\newcommand{\twoexp}{\textbf{\textsc{2Exptime}}}

\newtheorem{theorem}{Theorem}
\newtheorem{lemma}[theorem]{Lemma}

\newtheorem{corollary}[theorem]{Corollary}
\newtheorem{remark}[theorem]{Remark}
\newtheorem{construction}[theorem]{Construction}
\newtheorem{proposition}[theorem]{Proposition}

\title{Optimal Bounds in Parametric LTL Games
\thanks{The author's work was supported by the project
{\it Games for Analysis and Synthesis of Interactive Computational Systems
(GASICS)} of the {\it European Science Foundation}.}}

\author{Martin Zimmermann
\institute{Lehrstuhl Informatik 7\\
 RWTH Aachen University, Germany}
\email{zimmermann@automata.rwth-aachen.de}
}
\def\authorrunning{M. Zimmermann}
\def\titlerunning{Optimal Bounds in Parametric LTL Games}

\begin{document}
\maketitle

\begin{abstract}
We consider graph games of infinite duration with winning conditions
in parameterized linear temporal logic, where the temporal operators are equipped with
variables for time bounds. In model checking such specifications were introduced
as ``PLTL'' by Alur et al.\ and (in a different version called ``PROMPT-LTL'') by Kupferman et al..

We present an algorithm to determine optimal variable valuations that allow a player to win a game. Furthermore,
we show how to determine whether a player wins a game
with respect to some, infinitely many, or all valuations. All our
algorithms run in doubly-exponential time; so, adding bounded temporal operators does not
increase the complexity compared to solving plain LTL games.
\end{abstract}

\section{Introduction}

Many of todays problems in computer science are no longer concerned with programs that transform data and then terminate, but with non-terminating systems. Model-checking, the automated verification of closed systems (those that do not have to interact with an environment), is nowadays routinely performed in industrial settings. For open system (those that have to interact with a possibly antagonistic environment), the framework of infinite two-player games is a powerful and flexible tool to verify and synthesize such systems.
A crucial aspect of automated verification is the choice of a specification formalism, which should be simple enough to be used by practitioners without formal training in automata theory or logics. Here, \emph{Linear Temporal Logic} ($\ltl$) has turned out to be an expressive, but easy to use formalism: its advantages include a compact, variable-free syntax and intuitive semantics. For example, the specification ``every request $q$ is answered by a response $p$'' is expressed by $\phi=\G(q\rightarrow \F p)$.

However, $\ltl$ lacks capabilities to express timing constraints, e.g., it cannot express that every request is answered within an unknown, but fixed number of steps.  Also, in an infinite game with winning condition $\phi$, Player $0$ might have two winning strategies, one that answers every request within $m$ steps, and another one that takes $n$ steps, for some $n>m$. The first strategy is clearly preferable to the second one, but there is no guarantee that the first one is indeed computed, when the game is solved.

To overcome these shortcomings, several parameterized temporal logics~\cite{AE01, GLN10, KPV09} where introduced for the verification of closed systems: here one adds parametric bounds on the temporal operators. We are mainly concerned with \emph{Parametric Linear Temporal Logic} ($\pltl$)~\cite{AE01}, which adds the operators $\F_{\le x}$ and $\G_{\le y}$ to $\ltl$. In $\pltl$, the request-response specification is expressed by $\G(q\rightarrow \F_{\le x} p)$, stating that every request is answered within the next $x$ steps, where $x$ is a variable. Hence, satisfaction of a formula is defined with respect to a variable valuation $\alpha$ mapping variables to natural numbers: $\F_{\le x}\phi$ holds, if $\phi$ is satisfied
within the next $\alpha(x)$ steps, while $\G_{\le y}\phi$ holds, if $\phi$ is
satisfied for the next $\alpha(y)$ steps.

The model-checking problem for a parameterized temporal logic is typically no harder than the model-checking problem for the unparameterized fragment, e.g., deciding whether a transition system satisfies a $\pltl$ formula with respect to some, infinitely many, or all variable valuations is $\pspace$-complete~\cite{AE01}, as is $\ltl$ model-checking~\cite{SC85}. Similar results hold for parameterized real-time logics~\cite{GLN10}. Also, for $\pltl$ one can determine optimal variable valuations for which a formula is satisfied by a given transition system in polynomial space.

In this work, we consider infinite games with winning conditions in $\pltl$, i.e., we lift the results on model-checking parameterized specifications to synthesis of open systems from parameterized specifications. Our starting point is a result on the fragment of $\pltl$ containing only parameterized eventualities, which was discussed (in a different version called $\prompt$) in \cite{KPV09}: there, the authors show that the realizability problem (an abstract notion of a game) for $\prompt$ is $\twoexp$-complete. We use this result to solve infinite games with winning conditions in the full logic with parameterized eventualities and always': determining whether a player wins a $\pltl$ game with respect to some, infinitely many, or all variable valuations is also $\twoexp$-complete, as is determining the winner of an $\ltl$ game~\cite{PR89a}. So, we observe the same phenomenon as in model-checking: the addition of parameterized operators does not increase the computational complexity of the problem.

After establishing these results, we consider the problem of finding optimal variable valuations that allow a given player to win the game. If a winning condition contains only parameterized eventualities or only parameterized always', then it makes sense to ask for an optimal valuation that a player can enforce against her opponent and for a winning strategy realizing the optimum. Our main theorem states that this optimization problem can be solved in doubly-exponential time; so even determining an optimal winning strategy for such a game is of the same computational complexity as solving unparameterized games.

The remainder of this paper is structured as follows: in Section~\ref{sec_def}, we introduce infinite games with winning conditions in parameterized linear temporal logic and fix our notation. In Section~\ref{sec_solv}, we show how the result on the $\prompt$ realizability problem can be used to show that determining whether a player wins a $\pltl$ game with respect to some, infinitely many, or all variable valuations can be decided in doubly-exponential time. In Section~\ref{sec_opt}, we use these results to determine optimal winning strategies in games for which a notion of optimality can be defined. Finally, Section~\ref{sec_conc} gives a short conclusion.
\section{Definitions}
\label{sec_def}
The set of non-negative integers is denoted by $\nats$, the set of positive
integers by $\nats_+$. The powerset of a
set $S$ is denoted by $2^S$. Throughout this paper let $P$ be a set
of {atomic propositions}.\smallskip

%%%%%%%%%%%%%%%%%%%%%%%%%%%%%%%%%%%%%%%%%%%%%%%%%%%%%%%%%%%%%%%%%%%%%%%%%%%%%%%
%%%%%%%%%%%%%%%%%%%%%%%%%%%%%%%%%%%%%%%%%%%%%%%%%%%%%%%%%%%%%%%%%%%%%%%%%%%%%%%

\textbf{Automata.}
An {(non-deterministic) $\omega$-automaton} $\aut=(Q,\Sigma,Q_0,\Delta,\acc)$ consists of
a finite set of states $Q$, an alphabet $\Sigma$, a set of initial states $Q_0\subseteq Q$, a transition relation
$\Delta\subseteq Q\times\Sigma\times Q$, and an acceptance condition $\acc$.
An $\omega$-automaton is {deterministic}, if $|Q_0|=1$ and for every $(q,a)\in Q\times\Sigma$, there is exactly one $q'$
such that $(q,a,q')\in\Delta$. In this case, we denote $Q_0=\{q_0\}$ by $q_0$ and $\Delta$ as function $\delta\colon Q\times\Sigma\rightarrow Q$.
The size of $\aut$, denoted by $|\aut|$, is the cardinality of $Q$.
A {run} of $\aut$ on an $\omega$-word $w\in\Sigma^{\omega}$ is an infinite sequence of
states $q_0q_1q_2\dots$ such that $q_0\in Q_0$ and $(q_n,w_n,q_{n+1})\in\Delta$
for every $n\in\nats$.
We consider different {acceptance conditions} $\acc$ for
$\omega$-automata:
(1)~{B\"uchi automata} with a set of {accepting states} $F\subseteq Q$. A run $q_0q_1q_2\ldots$ is accepting
if there are infinitely many  $n$ such that $q_n\in F$.
(2)~{Generalized B\"uchi automata} with a family of sets of accepting states $\curlyF\subseteq 2^Q$. A run $q_0q_1q_2\ldots$ is accepting
if for every $F\in\curlyF$ there are infinitely many $n$ such that $q_n\in F$.
(3)~{Parity automata} with a {priority function} $c\colon Q\rightarrow \nats$. A run $q_0q_1q_2\ldots$ is accepting, if
%$\min(\infi(c(q_0)c(q_1)c(q_2)\ldots))$ is even, i.e.,
the minimal priority seen infinitely often is even.
An $\omega$-word is {accepted} by an $\omega$-automaton, if there exists an accepting run on it. The {language} $L(\aut)$ of $\aut$ contains the $\omega$-words accepted by $\aut$.
An $\omega$-automaton is called {unambiguous}, if it has at most one accepting run on every $\omega$-word $w\in\Sigma^{\omega}$.
It is called {non-confluent}, if for every $\omega$-word $w$ and two runs $q_0q_1q_2\dots$ and $q_0'q_1'q_2'\dots$ on $w$ we have for all $n$ that if $q_n=q_n'$, then $q_{m}=q_{m}'$ for every $m<n$. In a non-confluent $\omega$-automaton with $n$ states, every finite prefix of an $\omega$-word has at most $n$ finite runs, all of which can be uniquely identified by their last state.
Finally, a state of an $\omega$-automaton is {unproductive}, if it is not reachable from the initial state or if there is no accepting run starting from this state. Removing all unproductive states from a (generalized) B\"uchi or parity automaton does not change its language.

\begin{remark}
An unambiguous (generalized) B\"uchi or parity automaton without unproductive states is non-confluent.
\end{remark}\medskip

%%%%%%%%%%%%%%%%%%%%%%%%%%%%%%%%%%%%%%%%%%%%%%%%%%%%%%%%%%%%%%%%%%%%%%%%%%%%%%%
%%%%%%%%%%%%%%%%%%%%%%%%%%%%%%%%%%%%%%%%%%%%%%%%%%%%%%%%%%%%%%%%%%%%%%%%%%%%%%%
\textbf{Linear Temporal Logics.}
Let $\Xvar$ and $\Yvar$ be two disjoint sets of {variables}\footnote{If the sets of variables are not disjoint, already the model-checking problem for $\pltl$ is undecidable~\cite{AE01}.}. The formulae of
{Parametric Linear Temporal Logic} ($\pltl$)~\cite{AE01} are given by the
grammar
\[\phi\cceq p\mid\neg
p\mid\phi\wedge\phi\mid\phi\vee\phi\mid\X\phi\mid\phi\U\phi\mid\phi\R\phi\mid \F_{\le x}\phi\mid\G_{\le y}\phi\enspace,\]
where $p\in P$, $x\in\Xvar$ and $y\in\Yvar$. Also, we use the derived operators $\true\coloneq
p\vee\neg p$ and $\false\coloneq p\wedge\neg p$ for some fixed $p\in P$,
$\F\phi\coloneq\true\U\phi$, and $\G\phi\coloneq\false\R\phi$
\footnote{In~\cite{AE01}, the authors also introduced the operators $\U_{\le
x}$, $\R_{\le y}$, $\F_{> y}$, $\G_{> x}$, $\U_{> y}$, and $\R_{> x}$. However,
they showed that all these operators can be expressed using $\F_{\le x}$ and
$\G_{\le y}$ only, at the cost of a linear increase of the formula's size. Also,
we ignore constant bounds as they do not add expressiveness.}. The set of variables
occurring in $\phi$ is denoted by $\var(\phi)$ and defined in the obvious way.
The size $|\phi|$ of a formula $\phi$ is measured by counting the distinct subformulae of $\phi$.
We consider several fragments of $\pltl$: $\phi$ is an $\ltl$ formula, if
$\var(\phi)=\emptyset$; $\phi$ is a $\prompt$ formula~\cite{KPV09}, if $\var(\phi)$ is a
subset of $\Xvar$ of cardinality at most one; $\phi$ is a $\pltlf$ formula, if $\var(\phi)\subseteq \Xvar$; and
$\phi$ is a $\pltlg$ formula, if $\var(\phi)\subseteq \Yvar$. A formula in $\pltlf$ or $\pltlg$
is called {unipolar}.
The semantics of $\pltl$ is defined with
respect to an $\omega$-word $w\in\left(2^{P}\right)^{\omega}$, a position
$i\in\nats$, and a {variable valuation}
$\alpha\colon \Xvar\cup\Yvar\rightarrow\nats$ as follows:
% For the $\ltl$ operators, the
% semantics is standard (see~\cite{BK08}) and independent of $\alpha$. For the parameterized
% operators, we define
\begin{itemize}
\item $(w,i,\alpha)\models p$ iff $p\in w_i$ and $(w,i,\alpha)\models\neg p$ iff $p\notin
w_i$,
\item $(w,i,\alpha)\models\phi\wedge\psi$ iff $(w,i,\alpha)\models\phi$ and
$(w,i,\alpha)\models\psi$,
\item $(w,i,\alpha)\models\phi\vee\psi$ iff $(w,i,\alpha)\models\phi$ or $(w,i,\alpha)\models\psi$,
\item $(w,i,\alpha)\models\X\phi$ iff $(w,i+1,\alpha)\models\phi$,
\item $(w,i,\alpha)\models\phi\U\psi$ iff there exists a $j\ge 0$ such that
$(w,i+j,\alpha)\models\psi$ and $(w,i+j',\alpha)\models\phi$ for all $j'$ in the range $0\le j'<j$,
\item $(w,i,\alpha)\models\phi\R\psi$ iff for all $j\ge 0$: either $(w,i+j,\alpha)\models\psi$
or there exists a $j'$ in the range $0\le j'<j$ such that $(w,i+j',\alpha)\models\phi$,
\item $(w,i,\alpha)\models\F_{\le x}\phi$ iff there exists a $j$ in the range
$0\le j \le\alpha(x)$ such that $(w,i+j,\alpha)\models\phi$, and
\item $(w,i,\alpha)\models\G_{\le y}\phi$ iff for all $j$ in the
range $0\le j \le\alpha(y)$: $(w,i+j,\alpha)\models\phi$.
\end{itemize}
As the satisfaction of an $\ltl$ formula $\phi$ is independent of the variable valuation $\alpha$, we omit
$\alpha$ and write $(w,i)\models\phi$ instead of $(w,i,\alpha)\models\phi$.
$\pltl$ and $\ltl$ (but not the fragments $\prompt$, $\pltlf$ and $\pltlg$) are
closed under negation, although we only allow formulae in negation normal form.
This is due to the duality of $\U$ and $\R$, and $\F_{\le x}$ and $\G_{\le y}$.
Thus, we use $\neg\phi$ as shorthand for the equivalent formula obtained by
pushing the negation to the atomic propositions. Note that the negation of a $\pltlf$
formula is a $\pltlg$ formula and vice versa.

\begin{remark}
\label{rem_transl}
For every $\pltl$ formula $\phi$ and every valuation $\alpha$, there exists an
$\ltl$ formula $\phi_{\alpha}$ such that for every $w\in\left(2^{P}\right)^\omega$ and every $i\in\nats$: $(w,i,\alpha)\models\phi$ if and only if
$(w,i)\models\phi_{\alpha}$.
\end{remark}
This can be shown by replacing the parameterized operators by disjunctions or
conjunctions of nested next-operators. The size of $\phi_{\alpha}$
grows linearly  in $\sum_{z\in\var(\phi)}\alpha(z)$. Due to Remark~\ref{rem_transl}, we do not consider a fixed variable valuation when defining games with winning conditions in $\pltl$, but ask whether Player~$0$ can win a game with winning condition $\phi$
with respect to some, infinitely many, or all valuations.\medskip

%%%%%%%%%%%%%%%%%%%%%%%%%%%%%%%%%%%%
%%%%%%%%%%%%%%%%%%%%%%%%%%%%%%%%%%%%

\textbf{Infinite Games.}
An {(initialized and labeled) arena} $\mathcal{A}=(V,V_0,V_1,E,v_0,\ell)$ consists of a
finite directed graph $(V,E)$, a partition $\{V_0, V_1\}$ of $V$ denoting the
positions of {Player~$0$} and {Player~$1$},
an {initial vertex}
$v_0\in V$, and a {labeling function} $\ell\colon V\rightarrow 2^P$. The
size $|\mathcal{A}|$ of $\mathcal{A}$ is $|V|$. It is assumed
that every vertex has at least one outgoing edge. A {play}
$\rho=\rho_0\rho_1\rho_2\ldots$ is an infinite path starting in $v_0$. The
{trace} of $\rho$ is $t(\rho)=\ell(\rho_0)\ell(\rho_1)\ell(\rho_2)\ldots$.
A {strategy for Player~$i$} is a mapping
$\sigma\colon V^*V_i\rightarrow V$ such that $(\rho_n,\sigma(\rho_0\ldots\rho_n))\in
E$ for all $\rho_0\ldots\rho_n\in V^*V_i$. A play $\rho$ is
{consistent with $\sigma$} if $\rho_{n+1}=\sigma(\rho_0\ldots\rho_n)$ for
all $n$ with $\rho_n\in V_i$.

A {parity game} $\game=(\arena,c)$ consists of an arena $\arena$ and a priority function
$c\colon V\rightarrow\nats$. Player~$0$ wins a play $\rho_0\rho_1\rho_2\ldots$
if the minimal priority seen infinitely often is even. The number of priorities
of $\game$
is $|c(V)|$. A strategy $\sigma$ for Player~$i$ is winning for her, if every play
that is consistent with $\sigma$ is won by her. Then, we say Player~$i$ wins $\game$.

A {$\pltl$ game} $\game=(\arena,\phi)$ consists of an arena $\arena$ and a $\pltl$
formula~$\phi$. Player~$0$ {wins} a play $\rho$ {with respect to a variable
valuation $\alpha$} if $(t(\rho),0,\alpha)\models\phi$, otherwise Player~$1$ {wins
$\rho$ with respect to $\alpha$}. A strategy for Player~$i$ is a {winning strategy for her with respect to
$\alpha$} if every play that is consistent with $\sigma$ is won by Player~$i$
with respect to $\alpha$. Then, we say that Player~$i$ {wins $\game$ with respect
to $\alpha$}. We define
the set $\W_{\game}^i$ of {winning valuations} for Player~$i$ in $\game=(\arena,\phi)$ by
$\W_{\game}^i=\{\alpha\mid\text{Player~$i$ wins $\game$ with respect to
$\alpha$}\}$. Here (and from now on) we assume that $\alpha$'s domain
is restricted to the variables occurring in $\phi$. $\ltl$, $\prompt$, $\pltlf$, $\pltlg$, and
unipolar games are defined by restricting the winning conditions to $\ltl$, $\pltlf$,
$\pltlg$, and unipolar formulae. Again, winning an $\ltl$ game is independent of $\alpha$, hence we are justified to
say that {Player~$i$ wins an $\ltl$ game}.\medskip

\textbf{Strategies with Memory.} A {memory structure $\mem=(M,m_0,\update)$
for an arena $(V,V_0,V_1,E,v_0,\ell)$ consists of a finite set
$M$ of {memory states}, an {initial memory state} $m_0\in M$, and an
{update function} $\update\colon M\times V\rightarrow M$, which can
be extended to $\update^*\colon V^+\rightarrow M$ by
$\update^*(\rho_0)=m_0$ and
$\update^*(\rho_0\ldots\rho_n\rho_{n+1})=\update(\update^*(\rho_0\ldots\rho_n),
\rho_{n+1})$. A {next-move function for Player~$i$} is a function
$\nextmove\colon V_i\times M\rightarrow V$ that satisfies
$(v,\nextmove(v,m))\in E$ for all $v\in V_i$ and
all $m\in M$. It induces a {strategy $\sigma$ with memory $\mem$} via
$\sigma(\rho_0\ldots\rho_n)=\nextmove(\rho_n,\update^*(\rho_0\ldots\rho_n))$. A
strategy is called {finite-state} if it can be implemented with a
memory structure, and {positional} if it can be implemented with a single memory state.
The {size} of $\mem$ (and, slightly abusive, $\sigma$) is $|M|$.
An arena $\arena$ and a memory structure $\mem=(M,m_0,\update)$ for $\arena$
induce the expanded
arena $\arena\times\mem=(V\times M, V_0\times M, V_1\times M, E', (s_0, m_0),
\ell_{\arena\times\mem})$ where $((s,m),(s',m'))\in E'$ if and only if $(s,s')\in E$ and
$\update(m,s')=m'$, and $\ell_{\arena\times\mem}(s,m)=l(s)$.
%Every play $\rho'
%=(\rho_0,m_0)(\rho_1,m_1)(\rho_2,m_2)\ldots$ in $\arena\times\mem$ has a unique
%projected play $\rho=\rho_0\rho_1\rho_2\ldots$ in $\arena$.
A game $\game$ with
arena $\arena$ is reducible to $\game'$ with arena $\arena'$ via $\mem$, written
$\game\le_{\mem}\game'$, if $\arena'=\arena\times\mem$ and every play $(\rho_0,m_0)(\rho_1,m_1)(\rho_2,m_2)\ldots$ in
$\game'$ is won by the player who wins the projected play $\rho_0\rho_1\rho_2\ldots$ in $\game$.

\begin{remark}
\label{rem_reductions}
If $\game\le_{\mem}\game'$ and Player~$i$ has a positional winning
strategy for $\game'$, then she also has a finite-state winning strategy with
memory $\mem$ for $\game$.
\end{remark}

A parity game or a $\pltl$ game $\game$ (with respect to a fixed variable valuation) cannot
be won by both players. On the other hand, $\game$ is {determined}, if one of the players wins it.

\begin{proposition}
\label{prop_det}
\hfill
\begin{enumerate}
\item\label{prop_det_parity}
Parity games are determined with positional
strategies~\cite{EJ91, M91} and the winner can be determined in time $\bigo( m(n/d)^{\lceil d/2\rceil}) $~\cite{J00}, where $n$, $m$, and $d$ denote the number of vertices,
edges, and priorities.

\item\label{prop_det_ltl} $\ltl$ games (and therefore also $\pltl$ games with respect to a fixed variable
valuation) are determined with finite-state strategies.
Determining the winner is $\twoexp$-complete~\cite{PR89} and finite-state
winning strategies can be computed in doubly-exponential time.
\end{enumerate}
\end{proposition}

\section{Solving Prompt and PLTL Games}
\label{sec_solv}

In this section, we consider several decision problems for $\pltl$ games. Kupferman et al.\
solved the $\prompt$ realizability problem\footnote{An abstract game without underlying arena in which two players alternatingly pick letters from $2^{P}$. The first player wins if the $\omega$-word produced by the players satisfies the winning condition $\phi$.} by a reduction to the $\ltl$ realizability problem~\cite{KPV09}, which is complete for doubly-exponential time. We show that this result suffices to prove that even the
decision problems for the full logic with non-uniform bounds and parameterized always-operators are in $\twoexp$. 
For games with winning conditions in $\pltl$ we are interested in the following decision problems:\medskip

\noindent \textbf{Membership:} Given a  $\pltl$ game $\game$, $i\in\{0,1\}$, and a valuation $\alpha$, does $\alpha\in\W_{\game}^i$
hold?

\noindent \textbf{Emptiness:} Given a $\pltl$ game $\game$ and $i\in\{0,1\}$, is $\W_{\game}^i$ empty?

\noindent \textbf{Finiteness:} Given a $\pltl$ game $\game$ and $i\in\{0,1\}$, is $\W_{\game}^i$ finite?

\noindent \textbf{Universality:} Given a $\pltl$ game $\game$ and $i\in\{0,1\}$, does $\W_{\game}^i$ contain all variable valuations?\medskip

Our first result is a simple consequence of Remark~\ref{rem_transl} and Proposition~\ref{prop_det}.\ref{prop_det_ltl}.
% Let $\game=(\arena,\phi)$ be a $\pltl$ game and $\alpha$ a variable valuation.
% Applying Remark~\ref{rem_transl} to games, we obtain: Player~$i$ wins $\game$ with
% respect to $\alpha$ if and only if she wins the $\ltl$ game $(\arena,\phi_{\alpha})$. The winner
% and a finite-state winning strategy can be computed effectively due to Proposition~\ref{prop_det}.\ref{prop_det_ltl}.
%Hence, we
%obtain our first result about $\pltl$ games.

\begin{theorem}
\label{thm_memb}
The membership problem for $\pltl$ games is decidable.
\end{theorem}

The realizability problem for $\prompt$ is known to be $\twoexp$-complete. The proof of this result can easily be adapted
to graph-based $\prompt$ games as considered here.

\begin{theorem}[\cite{KPV09}]
\label{thm_prompt}
The emptiness problem for $\prompt$ games is $\twoexp$-complete.
\end{theorem}

The adapted proof in terms of graph-based games can be found in~\cite{Z10} and is sketched in the next section (see Lemma~\ref{lem_blink}). It
proceeds by a reduction to solving $\ltl$ games: given a $\prompt$ game
$\game=(\arena,\phi)$ one constructs an $\ltl$ game $\game'=(\arena',\phi')$ with
$|\arena'|\in\mathcal{O}(|\arena|^2)$ and
$|\phi'|\in\mathcal{O}(|\phi|)$ such that $\W_{\game}^0\not=\emptyset$ if and only if Player~$0$ wins $\game'$.
This proof yields the following corollary, which will be crucial when we determine optimal strategies in the next section: let $f(n)=2^{2^{75(n+1)}}\in 2^{2^{\bigo(n)}}$.

\begin{corollary}[\cite{Z10}]
\label{cor_pltlsize}
Let $\game=(\arena,\phi)$ be a $\prompt$ game with $\var(\phi)=\{x\}$. If $\W_\game^0\not=\emptyset$, then Player~$0$ also has a finite-state
winning strategy for $\game$ of size $2|\arena| f(|\phi|)$ which is winning with respect to the valuation $x\mapsto 2(|\arena|\cdot f(|\phi|)+1)$.
\end{corollary}

To solve the other decision problems for games with winning conditions in full $\pltl$,
we make use of the duality of unipolar games and the
duality of the emptiness and universality problem. For an arena
$\arena=(V,V_0,V_1,E,v_0,\ell)$, let $\dual{\arena}\coloneq (V,V_1,V_0,E,v_0,\ell)$ be its {dual
arena}, where the two players swap their positions. Given a $\pltl$ game
$\game=(\arena,\phi)$, the {dual game} is $\dual{\game}\coloneq (\dual{\arena},\neg\phi)$.
The dual game of a $\pltlg$ game is a $\pltlf$ game and vice versa. It is easy to see that
Player~$i$ wins $\game$ with respect to $\alpha$ if and only if Player~$1-i$ wins
$\dual{\game}$ with respect to $\alpha$.
The sets $\W_{\game}^i$ enjoy two types of dualities, which we rely on in the
following. The first one is due to determinacy of $\ltl$ games, the second one
due to duality.
\begin{lemma}
\label{lem_dual}
Let $\game$ be a $\pltl$ game.
\begin{enumerate}
\item\label{lem_dual_point1} $\W_{\game}^0$ is the complement of $\W_{\game}^1$.
\item\label{lem_dual_point2} $\W_{\game}^i=\W_{\dual{\game}}^{1-i}$.
\end{enumerate}
\end{lemma}
%

%short version
% \begin{lemma}
% \label{lem_dual}
% Let $\game$ be a $\pltl$ game. Then, $\W_{\game}^0$ is the complement of $\W_{\game}^1$ and $\W_{\game}^i=\W_{\dual{\game}}^{1-i}$.
% \end{lemma}

Another useful property is the monotonicity of the
parameterized operators: let $\alpha(x)\le \beta(x)$ and $\alpha(y)\ge
\beta(y)$. Then, $(w,i,\alpha)\models\F_{\le x}\phi$ implies
$(w,i,\beta)\models\F_{\le x}\phi$ and $(w,i,\alpha)\models\G_{\le y}\phi$
implies $(w,i,\beta)\models\G_{\le y}\phi$. Hence, the set $\W_{\game}^0$ is
upwards-closed if $\game$ is a $\pltlf$ game, and downwards-closed if $\game$ is
a $\pltlg$ game (valuations are compared componentwise).
% In the remainder, we solve the emptiness, finiteness, and universality problem for
% $\pltl$ games. We begin by considering unipolar games and then reduce the problems
% for $\pltl$ games to problems for unipolar games.
Now, we prove the main result of this section.
\begin{theorem}
\label{thm_pltl}
The emptiness, finiteness, and universality problems for $\pltl$ games
are $\twoexp$-complete.
\end{theorem}
\begin{proof}
Let $\game=(\arena,\phi)$. Due to Lemma~\ref{lem_dual}.\ref{lem_dual_point2}
it suffices to consider $i=0$.

\textbf{Emptiness of $\W_{\game}^0$:} Let $\phi_{\F}$ be the formula obtained
from $\phi$ by inductively replacing every subformula $\G_{\le y}\psi$ by
$\psi$, and let $\game_{\F}\coloneq(\arena,\phi_{\F})$. Note that $\game_{\F}$ is a $\pltlf$ game.
Applying downwards-closure, we obtain that $\W_{\game}^0$ is empty if and only if $\W_{\game_{\F}}^0$ is empty.

The latter problem can be decided by a reduction to $\prompt$ games. Fix a variable $x\in\Xvar$ and
let $\phi'$ be the formula obtained from $\phi_{\F}$ by replacing every variable $z$ in $\psi$ by $x$. Then,
$\W_{\game_{\F}}^0\not=\emptyset$ if and only if $\W_{\game'}^0\not=\emptyset$, where $\game'=(\arena, \phi')$.
The latter problem can be decided in doubly-exponential
time by Theorem~\ref{thm_prompt}. Since we have $|\phi'|\le |\phi|$, the emptiness of $\W_{\game}^0$ can be decided in doubly-exponential time.

\textbf{Universality of $\W_{\game}^0$:} Applying both statements of Lemma~\ref{lem_dual} we get that $\W_{\game}^0$ is universal if and only if $\W_{\game}^1=\emptyset$ if and only if $\W_{\dual{\game}}^0=\emptyset$.
The latter is decidable in doubly-exponential time, as shown above.
% 
%  Due to Lemma~\ref{lem_dual}%.\ref{lem_dual_point1}
% , $\W_{\game}^0$ is universal if and only if $\W_{\game}^1=\emptyset$. And, again due to Lemma~\ref{lem_dual}%.\ref{lem_dual_point2}
% , $\W_{\game}^1=\emptyset$ if and only if $\W_{\dual{\game}}^0=\emptyset$.
% The latter problem is decidable in doubly-exponential time, as shown above.

% Let $\phi_{\G}$ be the formula obtained
% from $\phi$ by inductively replacing every subformula $\F_{\le x}\psi$ by
% $\psi$, and let $\game_{\G}\coloneq(\arena,\phi_{\G})$. Note that $\game_{\G}$ is a
% $\pltlg$ game. Applying upwards-closure, we obtain
% that $\W_{\game}^0$ is universal if and only if $\W_{\game_{\G}}^0$ is universal.
% 
% The latter problem can be
% reduced to the emptiness problem for $\pltlf$ games. Applying
% Lemma~\ref{lem_dual} yields: $\W_{\game_{\G}}^0$ is universal if and only if
% $\W_{\dual{\game_{\G}}}^1$ is universal if and only if $\W_{\dual{\game_{\G}}}^0$ is empty. The
% latter problem is decidable in doubly-exponential time, as shown above.

\textbf{Finiteness of $\W_{\game}^0$:} If $\phi$ contains at least one $\F_{\le
x}$, then $\W_{\game}^0$ is infinite, if and only if it is non-empty, due to monotonicity
of $\F_{\le x}$. The emptiness of $\W_{\game}^0$ can be decided in
doubly-exponential time as discussed above.
Otherwise, $\game$ is a $\pltlg$ game whose finiteness problem
can be decided in doubly-exponential time by a reduction
to the universality problem for a (simpler) $\pltlg$ game. We assume that $\phi$
has at least one parameterized temporal operator, since the
problem is trivial otherwise. The set $\W_{\game}^0$ is infinite if and only if there is a
variable $y\in\var(\phi)$ that is mapped to infinitely many values by
the valuations in $\W_{\game}^0$. By downwards-closure we can assume that all
other variables are mapped to zero. Furthermore, $y$ is mapped to infinitely
many values if and only if it is mapped to all possible values, again by
downwards-closure. To combine this, we define $\phi_y$ to be the formula
obtained from $\phi$ by inductively replacing every subformula $\G_{\le z}\psi$
for $z\not=y$ by $\psi$ and define $\game_y\coloneq(\arena,\phi_y)$. Then,
$\W_{\game}^0$ is infinite, if and only if there exists some variable $y\in\var(\phi)$ such
that $\W_{\game_y}^0$ is universal. So, deciding whether $\W_{\game}^0$ is
infinite can be done in doubly-exponential time by solving $|\var(\phi)|$ many
universality problems for $\pltlg$ games, which were discussed above.

Finally, hardness follows directly from $\twoexp$-hardness of solving $\ltl$ games.
\end{proof}

\section{Optimal Winning Strategies for unipolar PLTL Games}
\label{sec_opt}
For unipolar games, it is natural to view synthesis of winning strategies as an
optimization problem: which is the {\it best} variable valuation $\alpha$ such
that Player~$0$ can win with respect to $\alpha$? We consider two quality
measures for a valuation $\alpha$ for $\phi$: the maximal parameter $\max_{ z\in
\var(\phi)} \alpha(z)$ and the minimal parameter $\min_{ z \in \var(\phi)}
\alpha(z)$. For a $\pltlf$ game, Player~$0$ tries to minimize the waiting times.
Hence, we are interested in minimizing the minimal or maximal parameter. Dually,
for $\pltlg$ games, we are interested in maximizing the quality measures. The dual
problems, i.e., maximizing the waiting times in a $\pltlf$ game and minimizing
the satisfaction time in a $\pltlg$ game, are trivial due to upwards- respectively downwards-
closure of the set of winning valuations. Again, we only consider Player~$0$ as one can dualize the game to obtain
similar results for Player~$1$. The main result of this section states that all
these op\-ti\-mi\-za\-tion problems are not harder than solving $\ltl$ games.

\begin{theorem}
\label{thm_pltlopt}
Let $\game_{\F}=(\arena_{\F},\phi_{\F})$ be a $\pltlf$ game and
$\game_{\G}=(\arena_{\G},\phi_{\G})$ be a $\pltlg$ game. Then, the following
values (and winning strategies realizing them) can be computed in
doubly-exponential time.
\begin{enumerate}
\item\label{thm_pltlopt1} $\min_{\alpha\in\W_{\game_{\F}}^0}\min_{x\in\var(\phi_{\F})}\alpha(x)$.
\item\label{thm_pltlopt2} $\min_{\alpha\in\W_{\game_{\F}}^0}\max_{x\in\var(\phi_{\F})}\alpha(x)$.
\item\label{thm_pltlopt3} $\max_{\alpha\in\W_{\game_{\G}}^0}\max_{y\in\var(\phi_{\G})}\alpha(y)$.
\item\label{thm_pltlopt4} $\max_{\alpha\in\W_{\game_{\G}}^0}\min_{y\in\var(\phi_{\G})}\alpha(y)$.
\end{enumerate}
\end{theorem}

% \begin{theorem}
% \label{thm_pltlopt}
% Let $\game_{\F}=(\arena_{\F},\phi_{\F})$ be a $\pltlf$ game and
% $\game_{\G}=(\arena_{\G},\phi_{\G})$ be a $\pltlg$ game. Then, the following
% values (and winning strategies realizing them) can be computed in
% doubly-exponential time.
% \begin{enumerate}
% \item\label{thm_pltlopt1}
% $\min_{\alpha\in\W_{\game_{\F}}^0}\min_{x\in\var(\phi_{\F})}\alpha(x)$.
% \item\label{thm_pltlopt2}
% $\min_{\alpha\in\W_{\game_{\F}}^0}\max_{x\in\var(\phi_{\F})}\alpha(x)$.
% \item\label{thm_pltlopt3}
% $\max_{\alpha\in\W_{\game_{\G}}^0}\max_{y\in\var(\phi_{\G})}\alpha(y)$.
% \item\label{thm_pltlopt4}
% $\max_{\alpha\in\W_{\game_{\G}}^0}\min_{y\in\var(\phi_{\G})}\alpha(y)$.
% \end{enumerate}
% \end{theorem}

%%%%%%%%%%%%%%%%%%%%%%%%%%%%%%%%%%%%%%%%%%%%%%%%%%%%%%%%%%%%%%%%%%%%%%%%%%
%%%%%%%%%%%%%%%%%%%%% begin comment %%%%%%%%%%%%%%%%%%%%%%%%%%%%%%%%%%%%%%
%%%%%%%%%%%%%%%%%%%%%%%%%%%%%%%%%%%%%%%%%%%%%%%%%%%%%%%%%%%%%%%%%%%%%%%%%%

We begin the proof by showing that all four problems can be reduced to the optimization problem for $\prompt$ games: let $\game=(\arena,\phi)$ be a $\prompt$ game with $\var(\phi)=\{x\}\subseteq \Xvar$. The goal is to determine $\min_{\alpha\in\W_{\game}^0} \alpha(x)$.  

The latter three reductions are simple applications of the monotonicity of the parameterized operators, while the first one requires substantial work. \smallskip

\ref{thm_pltlopt1}.) For each $x\in\var(\phi)$, we replace eventualities parameterized by $z\not=x$ by an unparameterized formula, thereby constructing the projection of $\W_{\game_\F}^0$ to the values of $x$. However, we cannot just remove the parameters from an eventuality, as we have to ensure that the waiting times are still bounded by some unknown, but fixed value. This is achieved by applying the alternating-color technique for $\prompt$~\cite{KPV09}.

Let $p\notin P$ be a fixed proposition. An $\omega$-word $w'=w_0'w_1'w_2'\ldots\in\left(2^{P\cup\{p\}}\right)^{\omega}$
is a $p$-coloring of $w=w_0w_1w_2\ldots\in\left(2^{P}\right)^{\omega}$ if
$w_n'\cap P=w_n$, i.e., $w_n$ and $w_n'$ coincide on all propositions in $P$.
The additional proposition $p$ can be thought of as the color of $w_n'$: we say
that a position $n$ is {green} if $p\in w_n'$, and say that it is {red}
if $p\notin w_n'$. Given $k\in\nats$ we say that $w'$ is {$k$-spaced},
if the colors in $w'$ change infinitely often, but not twice in any infix of
length $k$. Dually, $w'$ is {$k$-bounded}, if the colors change at least
once in every infix of length $k+1$.

The formula $alt_p\coloneq\G\F p\wedge \G\F\neg p$ is satisfied if the colors change
infinitely often. Given a $\pltl$ formula~$\phi$ and $X\subseteq var(\phi)$, let $\phi_X$ denote the formula obtained by inductively replacing every subformula $\F_{\le x}\psi$ with  $x\notin X$ by $(p\implies (p\U(\neg p\U\psi)))\wedge(\neg p\implies (\neg p\U( p\U\psi)))$.
Finally, consider the formula $\phi_X\wedge alt_p$. It forces a coloring to have infinitely many color changes and every subformula $\F_{\le x}\psi$ with $x\notin X$ to be satisfied within one color change. We have $\var(\phi_X) = X$ and $|\phi_X|\in \bigo(|\phi|)$.

For a variable valuation $\alpha$ and a subset $X$ of $\alpha$'s domain, we denote the restriction of $\alpha$ to $X$ by $\alpha_{\restriction X}$.

\begin{lemma}[\cite{KPV09}]
\label{lem_altcolor}
Let $\phi$ be a $\pltl$ formula, $X\subseteq\var(\phi)$, and let $w\in\left(2^P\right)^{\omega}$.
\begin{enumerate}
\item\label{lem_altcolor1} If $(w,0,\alpha)\models \phi$, then $(w',0,\alpha_{\restriction X})\models \phi_X\wedge alt_p$ for every $k$-spaced $p$-coloring $w'$ of $w$, where $k=\max_{x\in\var(\phi)\setminus X} \alpha(x)$.
\item\label{lem_altcolor2} Let $k\in\nats$. If $w'$ is a $k$-bounded $p$-coloring of $w$ with $(w',0,\alpha)\models \phi_X$, then $(w,0,\beta)\models\phi$ where $\beta(x)=\begin{cases}
\alpha(x) & \text{if $x\in X$,}\\
2k & \text{else.}                                                                                                                                                          \end{cases}
$
\end{enumerate}
\end{lemma}

The previous lemma shows how replace (on suitable $p$-colorings) a parameterized eventuality by an $\ltl$ formula, while still ensuring a bound on the satisfaction of the parameterized eventuality. To apply the alternating-color technique, we have to transform the original arena $\arena$ into an arena
$\arena'$ in which Player~$0$ produces $p$-colorings of the plays of the original
arena, i.e., $\arena'$ will consist of two disjoint copies of $\arena$, one labeled with
$p$, the other one not. Assume a play is in vertex $v$ in one component. Then, the
player  whose turn it is at $v$ chooses a successor $v'$ of $v$ and Player~$0$
picks a component. The play then continues in this component's vertex $v'$. We
split this into two sequential moves: first, the player whose turn it is chooses
a successor and then Player~$0$ chooses the component. Thus, we have to
introduce a new vertex for every edge of $\arena$ which allows Player~$0$ to choose
the component.
Formally, given an arena $\arena=(V,V_0,V_1,E,v_0,\ell)$, define the expanded arena $\arena'\coloneq (V',V_0',V_1',E',v_0',\ell')$ by
\begin{itemize}
\item $V'=V\times\{0,1\}\cup E$,
\item $V_0'=V_0\times\{0,1\}\cup E$,
\item $V_1'=V_1\times\{0,1\}$,
\item $E'=\{((v,0),e),((v,1),e),(e,(v',0)),(e,(v',1))\mid e=(v,v')\in E\}$,
\item $v_0'=(v_0,0)$,
\item $\ell'(e)=\emptyset$ for all $e\in E$ and
$\ell'(v,b)=\begin{cases}\ell(v)\cup\{p\} & \text{if $b=0$},\\
{\ell(v)}&\text{if $b=1$}.\\
         \end{cases}
$
\end{itemize}
$\arena'$ is bipartite with partition $\{V\times\{0,1\},E\}$, so a play
has the form $(\rho_0,b_0)e_0(\rho_1,b_1)e_1(\rho_2,b_2)\ldots$ where
$\rho_0\rho_1\rho_2\ldots$ is a play in $\arena$, $e_n=(\rho_n,\rho_{n+1})$, and the
$b_n$ are in $\{0,1\}$.
Also, we have $|\arena'|\in\bigo(|\arena|^2)$.

Finally, this construction necessitates a modification of the semantics
of the game: only every other vertex is significant when it
comes to determining the winner of a play in $\arena'$, the choice vertices
have to be ignored. This motivates {blinking semantics} for $\pltl$ games. Let
$\game=(\arena,\phi)$ be a $\pltl$ game and $\rho=\rho_0\rho_1\rho_2\ldots$ be a
play. Player~$0$ wins $\rho$ with respect to $\alpha$ under blinking semantics, if
$(t(\rho_0\rho_2\rho_4\ldots),0,\alpha)\models\phi$. Analogously, Player~$1$ wins
$\rho$ with respect to $\alpha$ under blinking semantics if
$(t(\rho_0\rho_2\rho_4\ldots),0,\alpha)\not\models\phi$. The notions of winning strategies and winning $\game$ with respect to $\alpha$ under blinking semantics are defined in the obvious way.

\begin{remark}
\label{rem_blink_fs}
$\pltl$ games with respect to a fixed variable valuation under blinking semantics are determined with finite-state strategies.
\end{remark}

Now, we can state the connection between a $\pltlf$ game $(\arena, \phi)$ and its counterpart in $\arena'$ with blinking semantics. The proof relies on the existence of finite-state winning strategies which necessarily produce only $k$-bounded plays for some fixed $k$, since $alt_p$ is part of the winning condition.

\begin{lemma}
\label{lem_blink}
Let $(\arena,\phi)$ be a $\pltlf$ game and $X\subseteq \var(\phi)$.
\begin{enumerate}
\item\label{lem_blink1} Let $\alpha\colon \var(\phi)\rightarrow\nats$ be a  variable valuation. If Player $i$ wins $(\arena, \phi)$ with respect to $\alpha$, then she wins $(\arena',\phi_X\wedge alt_p)$ with respect to $\alpha_{\restriction X}$ under blinking semantics.
\item\label{lem_blink2} Let $\alpha\colon X\rightarrow\nats$ be a  variable valuation. If Player $i$ wins $(\arena',\phi_X\wedge alt_p)$ with respect to $\alpha$ under blinking semantics, then there exists a variable valuation $\beta$ with $\beta(x)=\alpha(x)$ for every $x\in X$ such that she wins $(\arena, \phi)$ with respect to $\beta$.
\end{enumerate}
\end{lemma}

Applying the lemma to our problem, we have
\begin{equation*}
\min_{\alpha\in\W^0_\game}\min_{x\in\var(\phi)}\alpha(x)=\min_{x\in\var(\phi)}\min\{\alpha(x)\mid \text{ Player $0$ wins $(\arena',\phi_{\{x\}}\wedge alt_p)$}
\text{ w.r.t.\ $\alpha$ u. blinking semantics}\}\enspace.\end{equation*}
Since $\phi_{\{x\}}=\{x\}$, we have reduced the minimization problem to $|\var(\phi)|$ many $\prompt$ optimization problems, albeit under blinking semantics. However, the proof presented in the following can easily be adapted to deal with blinking semantics.

\ref{thm_pltlopt2}.) This problem can directly be reduced to a $\prompt$ optimization problem:
let $\phi_{\F}'$ be the $\prompt$ formula obtained from $\phi_{\F}$ by renaming
each $x\in\var(\phi_{\F})$ to $z$ and let $\game' \coloneq (\arena_{\F},
\phi_{\F}')$. Then, $\min_{ \alpha \in \W_{\game_{\F}}^0} \max_{
x\in\var(\phi_{\F})}\alpha(x)=\min_{ \alpha \in \W_{\game'}^0} \alpha(z)$, due
to up\-wards-closure of $\W_{\game_{\F}}^0$.

\ref{thm_pltlopt3}.) For every $y\in\var(\phi_{\G})$ let $\phi_y$ be obtained
from $\phi_{\G}$ by replacing every subformula $\G_{\le z}\psi$ for $z\not=y$ by
$\psi$ and let $\game_y\coloneq(\arena_{\G},\phi_y)$. Then, we have $\max_{
\alpha\in\W_{\game_{\G}}^0}\max_{y\in\var(\phi_{\G})}\alpha(y)=\max_{y \in
\var(\phi_{\G})}\max_{\alpha\in\W_{\game_{y}}^0}\alpha(y)$, due to
 downwards-closure of $\W_{\game_{\G}}^0$. Hence, we have reduced the original
problem to $|\var(\phi_{\G})|$ maximization problems for a $\pltlg$ game with
a single variable, which are discussed below.

\ref{thm_pltlopt4}.) Let $\phi_{\G}'$ be obtained from $\phi_{\G}$ by renaming
 every variable in $\phi_{\G}$ to $z$ and let $\game'=(\arena_{\G},\phi_{\G}')$.
Then, $\max_{ \alpha \in \W_{\game_{\G}}^0} \min_{ y \in \var(\phi_{\G})}
\alpha(y) = \max_{ \alpha \in \W_{\game'}^0} \alpha(z)$, again due to
downwards-closure of $\W_{\game_{\G}}^0$. Again, we have reduced the original
problem to a maximization problem for a $\pltlg$ game with a single variable.\smallskip

To finish the reductions we translate a $\pltlg$ optimization problem with a
single variable into a $\prompt$ optimization problem: let $\game=(\arena,\phi)$
be a $\pltlg$ game with $\var(\phi)=\{y\}\subseteq\Yvar$. Then, we have
$\max_{\alpha\in\W_{\game}^0}\alpha(y)=\max_{\alpha\in\W_{\dual{\game}}^1}
\alpha(y)=\min_{\alpha\in\W_{\dual{\game}}^0}\alpha(y)+1$,
due to the closure properties and Lemma~\ref{lem_dual}. As $\dual{\game}$ is a $\prompt$ game, we achieved
our goal.

%The reductions rely on the closure properties of the sets $\W_{\game}^i$ for unipolar games, the duality of $\pltlf$ and $\pltlg$ games, and the alternating-color technique~\cite{KPV09}. This technique allows to replace some parameterized eventualities by $\ltl$ subformulae. By changing the arena as well, it is ensured that Player~$0$ can satisfy the parameterized eventualities in the original game if and only if she can satisfy their $\ltl$ substitutes in the modified arena. Thus, for every $x\in\var(\phi)$ we can replace all but the eventualities parameterized by $x$, thereby obtaining a $\prompt$ game in each case.

All reductions increase the size of the arena at most quadratically and the size of the winning condition at most linearly. Furthermore, to minimize the minimal parameter value in a $\pltlf$ game and to maximize the maximal parameter value in a $\pltlg$ game, we have to solve $|\var(\phi)|$ many $\prompt$ optimization problems (for the other two problems just one) to solve the original unipolar optimization problem with winning condition $\phi$. Thus, it suffices to show that a $\prompt$ optimization problem can be solved in doubly-exponential time.

So, let $\game=(\arena,\phi)$ be a $\prompt$ game with $\var(\phi)=\{x\}$.
%%%%%%%%%%%%%%%%%%%%%%%%%%%%%%%%%%%%%%%%%%%%%%%%%%%%%%%%%%%%%%%%%%%%%%%%%%
%%%%%%%%%%%%%%%%%%%%% end comment %%%%%%%%%%%%%%%%%%%%%%%%%%%%%%%%%%%%%%%%
%%%%%%%%%%%%%%%%%%%%%%%%%%%%%%%%%%%%%%%%%%%%%%%%%%%%%%%%%%%%%%%%%%%%%%%%%%
%
% 
% Applying the closure-properties of the sets $\W_{\game_\F}^0$ and $\W_{\game_\G}^0$
% as well as the duality of $\pltlf$ and $\pltlg$ games all these problems can be reduced to
% the optimization problem for $\prompt$ games: let $\game=(\arena,\phi)$ be a $\prompt$ game,
% i.e., $\var(\phi)=\{x\}\subseteq \Xvar$. The goal is to determine $\min_{\alpha\in\W_{\game}^0}
% \alpha(x)$ and a winning strategy realizing this value. 
% All but one reduction do not increase the size of the winning condition (the increase is linear in the exceptional case)
% and no reduction does increase the size of the arena.\medskip
% 
If $\W_{\game}^0\not=\emptyset$, then Corollary~\ref{cor_pltlsize}
yields $\min_{ \alpha \in \W_{\game}^0}
\alpha(x) \le\;  k\;\coloneq 2(|\arena| \cdot f(|\phi|)+1)\in
|\arena|\cdot 2^{2^{\bigo(|\phi|)}}$.
Let $\alpha_{n}$ be the valuation mapping $x$ to $n$. To determine
$\min_{ \alpha \in \W_{\game}^0}\alpha(x)$, it suffices to find the smallest $n<k$ such that
$\alpha_{n}\in\W_{\game}^0$. As the number
of such valuations $\alpha_{n}$ is equal to $k$,
it suffices to show that $\alpha_{ n } \in \W_{ \game}^0$ can be
decided in doubly-exponential time in the size of $\game$,
provided that $n< k$. This is achieved by a game reduction to a parity
game.

Fix a valuation $\alpha$ and remember that $\phi_{\alpha}$ is an $\ltl$ formula (see Remark~\ref{rem_transl}). Now, observe that a deterministic parity automaton
$\autp=(Q,2^P,q_0,\delta,c)$ with $L(\autp)=\{w\in (2^{P})^{\omega}
\mid (w,0) \models \phi_{\alpha}\}$ can be turned into a memory structure
$\mem=(Q,q_0,\update)$ for $(\arena, \phi_{\alpha})$ by defining
$\update(q,v) = \delta(q, \ell(v))$. Then, we have $(\arena, \phi_{ \alpha})
\le_{ \mem }( \arena \times \mem, c')$, where $c'(v,q)=c(q)$. Hence,
the Remarks \ref{rem_transl} and \ref{rem_reductions} yield $\alpha\in\W_{\game}^0$ if and only if Player~$0$ wins $(\arena\times\mem,c')$.

\begin{lemma}
\label{lem_autsize}
Let $\alpha$ be a variable valuation and $\phi$ be a $\prompt$ formula with $\var(\phi)=\{x\}$.
There exists a deterministic parity automaton $\autp$ recognizing the language $\{w\in (2^{P})^{\omega} \mid (w,0) \models \phi_\alpha\}$ such that
$|\autp| \in 2^{2^{\bigo(|\phi|)}}\cdot ( \alpha(x)+1 )^{2^{\bigo(|\phi|)}}$
and $\autp$ has $2^{\bigo(|\phi|)}$ many colors.
\end{lemma}

For a valuation $\alpha_{n}$ with $n < k$, we have $|\autp|\in 2^{2^{\bigo(|\arena|+|\phi|)}}$ with $2^{\bigo(|\phi|)}$ many colors. Thus, Proposition~\ref{prop_det}.\ref{prop_det_parity} implies that $( \arena \times \mem, c')$ can be solved in doubly-exponential time in the size of $\game$, which suffices to prove Theorem~\ref{thm_pltlopt}, as we have to solve at most doubly-exponentially many parity games\footnote{This can be improved to exponentially many by binary search.}, each of which can be solved in doubly-exponential time. Thus, it remains to prove Lemma~\ref{lem_autsize}.

Furthermore, we have seen that the automaton $\autp$ for the \textit{minimal} $\alpha_{n}$ can easily be turned into a finite-state winning strategy for $\game$ realizing $\min_{\alpha\in\W_{\game}^0}\alpha(x)$. To obtain a winning strategy for the general case of an $\pltlf$ (respectively $\pltlg$) game it is necessary to construct a deterministic parity automaton for the $\pltlf$ formula $\phi$ (respectively $\neg\phi$) as described below. In case of a $\pltlg$ game, we need to complement the automaton, which is achieved by incrementing the priority of each state by one.

We construct an automaton as required in Lemma~\ref{lem_autsize} in the remainder of this section. Note that the naive approach of constructing a deterministic parity automaton for the $\ltl$ formula $\phi_{\alpha_n}$ yields an automaton that recognizes the desired language,
but is of quadruply-exponential size, if $n$ is close to $k$. The problem arises from the fact that $\phi_{\alpha_n}$ uses a disjunction of nested
next-operators of depth $n$ to be able to count up to
$n$. This (doubly-exponential)
\textit{counter} is hardwired into the formula $\phi_{\alpha_{n}}$ and
thus leads to a quadruply-exponential blowup when turning $\phi_{\alpha_{n}}$ into
a deterministic parity automaton, since turning $\ltl$ formulae into deterministic parity
automata necessarily incurs a doubly-exponential blowup~\cite{KV98}.

To obtain our results, we decouple the counter from the formula by
relaxing parameterized eventualities to plain eventualities. We translate the relaxed formula into a generalized B\"uchi automaton, which is then turned in a B\"uchi automaton.
By placing an additional constraint on accepting runs we take care of the bound on
the (now relaxed) parameterized operators. As these automata are unambiguous, we also end up with a non-confluent B\"uchi automaton, which is then determinized into a parity automaton. Only then, the additional constraint is added to the parity automaton in the form of a counter that tracks (and aborts, if the counter is overrun) different runs of the B\"uchi automaton. This way, we obtain an automaton that is equivalent to the (unrelaxed) $\prompt$ formula with respect to $\alpha_{n}$. To add these counters, it is crucial to have a non-confluent B\"uchi automaton, as such an automaton has at most $|Q|$ runs which have to be tracked by the counter.

In the following we extend known constructions for translating an $\ltl$ formula into a non-deter\-mi\-nistic B\"uchi automaton and for translating a non-deterministic B\"uchi automaton into a deterministic parity automaton. In the first step we have to deal with the additional constraints, which do not appear in the classical translation problem. In the second step, we have to simulate these constraints with the states of the parity automaton, which requires changes to this translation as well. Since our proof technique can deal with several parameters, we consider the more general case of a $\pltlf$ formula instead of a $\prompt$ formula.\smallskip

\noindent\textbf{From $\pltlf$ to generalized B\"uchi Automata.}
We begin by constructing a generalized B\"uchi automaton from a $\pltlf$ formula
using a slight adaptation of a standard textbook method (see~\cite{BK08}). We ignore
the parameters when defining the transition relation, i.e., we treat a
parameterized eventually as a plain eventually. The bounds are taken care of
by additional constraints on accepting runs.
%\begin{definition}

Given a $\pltlf$ formula $\phi$ we define its closure $\cl(\phi)$ to be the set
of subformulae of $\phi$.
% inductively by
% \begin{itemize}
% \item $\cl(p)=\{p\}$,
% \item $\cl(\neg p)=\{\neg p\}$,
% \item $\cl(\psi_1\wedge\psi_2)=\cl(\psi_1)\cup
% \cl(\psi_2)\cup\{\psi_1\wedge\psi_2\}$,
% \item $\cl(\psi_1\vee\psi_2)=\cl(\psi_1)\cup
% \cl(\psi_2)\cup\{\psi_1\vee\psi_2\}$,
% \item $\cl(\X\psi_1)=\cl(\psi_1)\cup\{\X\psi\}$,
% \item $\cl(\psi_1\U\psi_2)=\cl(\psi_1)\cup \cl(\psi_2)\cup\{\psi_1\U\psi_2\}$,
% \item $\cl(\psi_1\R\psi_2)=\cl(\psi_1)\cup \cl(\psi_2)\cup\{\psi_1\R\psi_2\}$,
% and
% \item $\cl(\F_{\le x}\psi_1)=\cl(\psi_1)\cup\{\F_{\le x}\psi_1\}$.
% \end{itemize}
A set $B\subseteq \cl(\phi)$ is consistent, if the following properties are
satisfied:

\begin{center}
\begin{tabular}{ll}
\begin{minipage}{0.5\linewidth}
\begin{itemize}
\item $p\in B$ if and only if $\neg p\notin B$ for every $p\in P$.
\item $\psi_1\wedge\psi_2\in B$ if and only if $\psi_1\in B$ and $\psi_2\in B$.
\item $\psi_1\vee\psi_2\in B$ if and only if $\psi_1\in B$ or $\psi_2\in B$. 
\end{itemize}
\end{minipage}

&

\begin{minipage}{0.4\linewidth}
\begin{itemize}
\item $\psi_2\in B$ implies $\psi_1\U\psi_2\in B$.
\item $\psi_1,\psi_2\in B$ implies $\psi_1\R\psi_2\in B$.
\item $\psi_1\in B$ implies $\F_{\le x}\psi_1\in B$.
\end{itemize}
\end{minipage}
\\
\end{tabular}
\end{center}
The set of consistent subsets is denoted by $\curlyC(\phi)\subseteq
2^{\cl(\phi)}$.
%\end{definition}

\begin{construction}\label{cons_gnba}
Given a $\pltlf$ formula $\phi$, we define the generalized B\"uchi automaton\newline
$\aut_{\phi} = (Q, 2^P, Q_0, \Delta, \curlyF)$ by
\begin{itemize}
\item $Q=\curlyC(\phi)$ and $Q_0=\{B\in \curlyC(\phi)\mid \phi\in B\}$,
\item $(B,a,B')\in\Delta$ if and only if
\begin{itemize}
\item $B\cap P=a$,
\item $\X\psi_1\in B$ if and only if $\psi_1\in B'$,
\item $\psi_1\U\psi_2\in B$ if and only if $\psi_2\in B$ or ($\psi_1\in B$ and
$\psi_1\U\psi_2\in B'$),
\item $\psi_1\R\psi_2\in B$ if and only if $\psi_2\in B$ and ($\psi_1\in B$ or
$\psi_1\R\psi_2\in B'$), and
\item $\F_{\le x}\psi_1\in B$ if and only if $\psi_1\in B$ or $\F_{\le x}\psi_1\in B'$.
\end{itemize}
\item $\curlyF=\curlyF_{\U}\cup\curlyF_{\R}\cup\curlyF_{\F_{\le}}$ where
\begin{itemize}

\item $\curlyF_{\U}=\{F_{\psi_1\U\psi_2}\mid \psi_1\U\psi_2\in\cl(\phi)\}$ with
$F_{\psi_1\U\psi_2}=\{B\in \curlyC(\phi)
\mid \psi_1\U\psi_2\notin B \text{ or } \psi_2\in B\}$,

\item $\curlyF_{\R}=\{F_{\psi_1\R\psi_2}\mid \psi_1\R\psi_2\in\cl(\phi)\}$ with
$F_{\psi_1\R\psi_2}=\{B\in \curlyC(\phi)
\mid \psi_1\R\psi_2\in B \text{ or } \psi_2\notin B\}$, and

\item $\curlyF_{\F_{\le}}=\{F_{\F_{\le x}\psi_1}\mid \F_{\le
x}\psi_1\in\cl(\phi)\}$ with $F_{\F_{\le x}\psi_1}=\{B\in \curlyC(\phi)
\mid \F_{\le x}\psi_1\notin B \text{ or } \psi_1\in B\}$.

\end{itemize}
\end{itemize}
\end{construction}\smallskip

\begin{lemma}\label{lem_gnba}
Let $\phi\in\pltlf$ and let $\aut_{\phi}$ be defined as in
Construction~\ref{cons_gnba}.
\begin{enumerate}
\item\label{lem_gnba1} 
%Let $w\in \left(2^{P}\right)^{\omega}$ and let $\alpha$ be a variable valuation. Then, 
$(w,0,\alpha)\models\phi$ if and only if $\aut_{\phi}$ has an
accepting run $\rho$ on $w$ such that each $F_{\F_{\le
x}\psi_1}\in\curlyF_{\F_{\le}}$ is visited at least once in
every infix of $\rho$ of length $\alpha(x)+1$.
\item\label{lem_gnba2} $\aut_{\phi}$ is unambiguous.
\item\label{lem_gnba3} $|\aut_{\phi}|\le 2^{|\phi|}$ and $|\curlyF|<|\phi|$.
\end{enumerate}
\end{lemma}

\begin{proof}
\ref{lem_gnba1}.)
Let $(w,0,\alpha)\models\phi$. For each $n$ define $B_n=\{\psi\in\cl(\phi)\mid
(w,n,\alpha)\models\psi\}$ and show that $\rho=B_0B_1B_2\dots$ is an accepting
run of $\aut_{\phi}$ such that each $F_{\F_{\le x}\psi_1}\in\curlyF_{\F_{\le}}$
is visited at least once in every infix of $\rho$ of length $\alpha(x)+1$.
The semantics of $\pltl$ guarantee that each $B_n$ is consistent, $B_0\in Q_0$
follows from $(w,0,\alpha)\models\phi$, and $(B_n,w_{n},B_{n+1})\in\Delta$ for
every $n$ is due to the semantics of $\pltl$. Thus, the sequence
$B_0B_1B_2\dots$ is a run.
Assume that some $F_{\psi_1\U\psi_2}$ is visited only finitely often, i.e.,
there exists an index~$n$ such that for every $n'\ge n$ we have
$(w,n',\alpha)\models\psi_1\U\psi_2$ and $(w,n',\alpha)\not\models\psi_2$. This
contradicts the semantics of the until-operator, which guarantee a position
$m\ge n$ such that $(w,m,\alpha)\models\psi_2$, if
$(w,n,\alpha)\models\psi_1\U\psi_2$.
Now, assume that some $F_{\psi_1\R\psi_2}$ is visited only finitely often, i.e.,
there exists an index~$n$ such that for every $n'\ge n$ we have
$(w,n',\alpha)\not\models\psi_1\R\psi_2$ and $(w,n',\alpha)\models\psi_2$. This
contradicts the semantics of the release-operator, which state
$(w,n,\alpha)\models\psi_1\R\psi_2$, if $\psi_2$ holds at every position $n'\ge
n$.
Finally, assume that some $F_{\F_{\le x}\psi_1}\in\curlyF_{\F_{\le}}$ is not
visited in an infix of $B_0B_1B_2\dots$ of length $\alpha(x)+1$, i.e., there is
some index $n$ such that $(w,n,\alpha)\models\F_{\le x}\psi_1$ and
$(w,n+j,\alpha)\not\models\psi_1$ for every $j$ in the range $0\le j\le
\alpha(x)$. This contradicts the semantics of the parameterized eventually,
which guarantee the existence of an index~$k$ in the range $0\le k \le\alpha(x)$ such that
$(w,n+k,\alpha)\models\psi_1$.
Hence, $B_0B_1B_2\dots$ is an accepting run such that each $F_{\F_{\le
x}\psi_1}\in\curlyF_{\F_{\le}}$ is visited at least once in every infix of
$B_0B_1B_2\dots$ of length $\alpha(x)+1$.

For the other direction, let $\rho=B_0B_1B_2\ldots$ be an accepting run of
$\aut_{\phi}$ on $w$ such that each $F_{\F_{\le x}\psi_1}\in\curlyF_{\F_{\le}}$
is visited at least once in
every infix of $\rho$ of length $\alpha(x)+1$. A structural induction over
$\phi$ shows that $\psi\in B_n$ if and only if $(w,n,\alpha)\models\psi$. This suffices,
since we have $\phi\in B_0$.

\ref{lem_gnba2}.)
Let $e(\phi)$ be the formula obtained from $\phi\in\pltlf$ by replacing every
parameterized eventually $\F_{\le x}$ by an eventually $\F$. The
automata $\aut_{\phi}$ and $\aut_{e(\phi)}$ are isomorphic. Thus, it suffices
to show that $\aut_{e(\phi)}$ is unambiguous. So, assume there are two
accepting runs $B_0B_1B_2\dots$ and $B_0'B_1'B_2'\dots$ on an
$\omega$-word $w$ and let $n$ be an index such that $B_n\not=B_n'$, i.e., there
exists $\psi\in\cl(e(\phi))$ such that (w.l.o.g.) $\psi\in B_n$, but $\psi\notin
B_n'$.
In~\ref{lem_gnba1}.), we have shown that we have $\psi\in B_n$ (respectively
$\psi\in B_n'$) if and only if $(w,n)\models\psi$ (note that $\psi$ is an $\ltl$ formula,
hence we do not need to care about a variable valuation). Thus, we have
$(w,n)\models\psi$ (due to $\psi\in B_n$) and $(w,n)\not\models\psi$ (due to
$\psi\not\in B_n'$), which yields the desired contradiction.

\ref{lem_gnba3}.) Clear.
\end{proof}

% %%%%%%%%%%%%%%          %%%%%%%%%%%%%%%%%%
% %%%%%%%%%%%%%%	  %%%%%%%%%%%%%%%%%%

\noindent\textbf{From generalized B\"uchi Automata to B\"uchi Automata.}
Now, we use a standard construction (see~\cite{BK08}) to turn a generalized B\"uchi
automaton $\aut=(Q,\Sigma,Q_0,\Delta,\{F_1,\ldots,F_k\})$ into a B\"uchi
automaton $\aut'=(Q',\Sigma,Q_0',\Delta',F')$ while preserving its language
(even under the additional constraints) and its unambiguity. The state set of
$\aut'$ is $Q\times \{0,1,\ldots,k\}$, where the first component is used to
simulate the behavior of $\aut$, while the second component is used to ensure
that every set $F_j$ is visited infinitely often.

\begin{lemma}\label{lem_nba}
Let $\aut=(Q,\Sigma,Q_0,\Delta,\{F_1,\ldots,F_k\})$ be a generalized B\"uchi
automaton. There exists a B\"uchi automaton $\aut'$ with state set
$Q\times\{0,1\dots,k\}$ such that the following holds:
\begin{enumerate}
%%\item\label{lem_nba1} $L(\aut)=L(\aut')$.
\item\label{lem_nba4}  Let $\aut=\aut_{\phi}$ for some $\pltlf$ formula $\phi$
as in Construction~\ref{cons_gnba}. Then, $(w,0,\alpha)\models\phi$ if and only if $\aut'$
has an accepting run $(q_0,i_0)(q_1,i_1)(q_2,i_2)\ldots$ on $w$ such that each
$\F_{\F_{\le x}\psi_1}\in\curlyF_{\F_{\le}}$ is visited at least once in every
infix of $q_0q_1q_2\dots$ of length $\alpha(x)+1$.

\item\label{lem_nba2} $\aut'$ is unambiguous, if $\aut$ is unambiguous.
\item\label{lem_nba3} $|\aut'|=|\aut|\cdot(k+1) $
\end{enumerate}
\end{lemma}

\noindent\textbf{From B\"uchi Automata to Deterministic Parity Automata.}
Now, we have to determinize an unambiguous (and therefore non-confluent) B\"uchi
automaton while incorporating the additional constraints on accepting runs.
Abstractly,
we are  given a non-confluent B\"uchi automaton $\aut$ and a finite set of tuples
$(F_j,b_j)\in 2^{Q}\times\nats_+$ and are only interested in runs
$\rho$ that visit a state from $F_j$ in every infix of $\rho$ of length $b_j$,
while visiting the accepting states of the B\"uchi
automaton infinitely often. Remember that a non-confluent automaton has at
most $|Q|$ finite runs on a finite word $w_0\cdots w_n$, which can be
 uniquely identified by their last state. Furthermore, for every last state $q$
of such a run, there is a unique state $p$ such that $p$ is the last state of a run
of the automaton on $w_0\cdots w_{n-1}$ and $(p,w_n,q)\in\Delta$. Thus, to
check the additional constraints on the runs, we can use counters $d(q,j)$ to
abort the run ending in $q$ if it did not visit $F_j$ for $b_j$ consecutive
states. The state space of the deterministic automaton we construct is the cartesian product of the state space of $\autp$ and the counters $d(q,j)$ for every $q$ and $j$, where $\autp$ is a deterministic automaton recognizing the language of $\aut$ without additional constraints. To prove Theorem~\ref{thm_pltlopt}, we want to use the deterministic automaton with counters as memory structure in a game reduction, which imposes additional requirements on its size and its acceptance condition.

The B\"uchi automaton we need to determinize is already of exponential size. Hence, we can spend another exponential for determinization, which is the typical complexity of a determinization procedure for B\"uchi automata. However, we have to carefully choose the acceptance condition of the deterministic automaton we construct: to prove the main theorem, we need an acceptance condition $\acc$ such that a game  with arena $\arena\times\mem$ and winning condition $\acc$ can be solved in doubly-exponential time, even if $\mem$ is already of doubly-exponential size. Furthermore, it is desirable to use a condition $\acc$ that guarantees Player~$0$ positional winning strategies: in this case, $\mem$ implements a finite-state winning strategy for her in the original $\pltlf$ game.

The parity condition satisfies all our requirements. Thus, we adapt a determinization construction~\cite{M10, MS08} tailored for non-confluent B\"uchi automata yielding a parity automaton. The automata obtained by this construction are slightly larger than the ones obtained by optimal constructions, but still small enough to satisfy our requirements on them. Another advantage of this construction is the fact that it is conceptually simpler than the constructions for arbitrary B\"uchi automata based on trees labeled with state sets. Nevertheless, it is possible to use another determinization construction, as long as it satisfies the requirements in terms of size and winning condition described above. 

Given a transition relation $\Delta\subseteq Q \times  \Sigma \times Q$, define
$\Delta(S,a)=\{q'\in Q\mid(q,a,q')\in\Delta \text{ for some }q\in S\}$.

\begin{construction}[\cite{MS08}]\label{cons_countpar}
Given a non-confluent B\"uchi automaton $\aut=(Q,\Sigma,Q_0,\Delta,F)$ and a
finite set $\{(F_1,b_1),\dots,(F_k,b_k)\}\subseteq 2^{Q} \times \nats_+$
construct the deterministic parity automaton $\autp=(Q',\Sigma,q_0',\delta,c)$
as follows: let $n=|Q|$ and define
\begin{itemize}
\item $Q'=\{((S_0,m_0),\ldots,(S_{n},m_{n}),d)\mid S_i\subseteq Q,\,
m_i\in\{0,1\},\, d\colon
Q\times\{1,\ldots,k\}\rightarrow\nats\cup\{\bot\} \text{ with } d(q,j)< b_j\text{ or }
d(q,j)=\bot\}$,
\item $q_0'=((S_0,0),(\emptyset,0),\dots,(\emptyset,0),d_0)$ with $d_0(q,j)=0$ if $q\in Q_0 \cap F_j$; $d_0(q,j)=1$ if $q\in Q_0 \setminus F_j$ and $1 < b_j$; and $d_0(q,j)= \bot$ otherwise; and $S_0=\{q\in Q_0\mid d(q,j)\not=\bot\text{ for every j}\}$.
\item We define the transition function $\delta$ only for reachable states:
$\delta(((S_0,m_0),\dots,(S_{n},m_{n}),d),a)=((S_0',m_0'),\dots,(S_{n}',m_{n}
'),d')$ where
\begin{itemize}
\item $d'(q,j)=\begin{cases}
0&  \text{if }q\in\Delta(S_0,a)\text{ and } q\in F_j,\\
d(p,j)+1&  \text{if }q\in\Delta(S_0,a)\text{, } q\notin F_j\text{, and
}d(p,j)+1<b_j,\\
\bot&  \text{if }q\in\Delta(S_0,a)\text{, } q\notin F_j\text{, and
}d(p,j)+1=b_j,\\
\bot&\text{if }q\notin\Delta(S_0,a),\\
\end{cases}$\enspace\\
where $p$ is the unique (due to non-confluence, see
Lemma~\ref{lem_parrunprop}.\ref{lem_parrunprop1}) state in $S_0$ with $(p,a,q)\in\Delta$.
Define $T=\{q\in Q\mid d'(q,j)\not=\bot\text{ for every }j\}$.

\item For the update of the state sets consider the sequence
$(S_0,m_0),\dots,(S_n,m_n)$ as a list containing tuples $(S,m)$. Remark~\ref{rem_parstateprop}.\ref{rem_parstateprop2} yields that there are at most $n$ non-empty sets $S_i$. First, we delete
all elements of the list containing the empty set by moving the non-empty state
sets to the left, without changing their order. Then, we replace every $S_i$ by
$\Delta(S_i,a)\cap T$. Finally, we append the state set $S_0\cap F$ to the end of
the list. Denote the length of the updated list by $\ell$.
Now, we clean up states. For $i=0,\dots,\ell-1$ do: if $S_i\setminus F$ is
a subset of $\bigcup_{i'=i+1}^{\ell-1} S_{i'}$ and $S_i\not=\emptyset$, then set $m_i'=1$, otherwise $m_i'=0$. Now, if $m_i=1$, then
remove the states contained in $S_i$ from every $S_{i'}$ with $i'>i$.
As we have $\ell\le n+1$, we can retranslate the updated list into a unique state tuple $((S_0',m_0'),\dots,(S_{n}',m_{n}'))$ (if the list is
too short, we pad it with $(\emptyset,0)$ at the end).
\end{itemize}

\item To define $c$ consider a reachable state $q=((S_0,m_0),\dots,(S_{n},m_{n}),d)$.
Let $e$ be the minimal $i$ such that $S_i=\emptyset$ and let $m$ be the minimal
$i$ such that $m_i=1$. Note that $e$ is always defined for reachable states (due to Remark~\ref{rem_parstateprop}.\ref{rem_parstateprop2}) and that $e\not=m$. We define

$c(q)=\begin{cases}
1&\text{if } e=0,\\
2m&\text{if } m<e,\\
2e-1&\text{if } 0<e< m\text{ or if $m$ undefined}.\\
\end{cases}$
\end{itemize}
\end{construction}
Note that in the definition of $\delta$, cleaning up the sets might introduce
new empty sets in the middle of the list. Also, note that $p$ in the definition of $d'$ is
only well defined when considering reachable states.
To prove the correctness of this construction, we need some properties of the
states of $\autp$.

\begin{remark}\label{rem_parstateprop} Let $q' = (( S_0,m_0), \dots,
(S_{n},m_{n}), d)$ be a reachable state of $\autp$.
\begin{enumerate}
\item\label{rem_parstateprop1} $S_i\subseteq S_0$ for every $i$.
\item\label{rem_parstateprop2} For every non-empty set $S_i$ there is a state $q_i\in S_i$ such that $q_i\notin S_{i'}$ for every $i'>i$. 
\item\label{rem_parstateprop3} $S_0=\{q\in Q\mid d(q,j)\not=\bot\text{ for
every }j\}$.
\end{enumerate}
\end{remark}

% \marginpar{delete me}
% \begin{proof}
%
% \ref{rem_parstateprop1}.) By induction over the set of reachable states: the
% claim clearly holds for the initial state. Now, let $q''$ be a state of $\autp$ for which
% the claim holds and consider $q'=\delta(q'',a)$. The sets $S_i''$ of $q''$
% are replaced by $\Delta(S_i'',a)\cap T$ which preserves the subset relation.
% Then, the set $(\Delta(S_0'',a)\cap T)\cap F$ is added, which is again a subset of $S_0=\Delta(S_0'',a)\cap T$.
% Finally, some states in the sets $S_i$ of $q'$ (for some $i>0$) may be deleted,
% which obviously preserves the subset relation.
%
% \ref{rem_parstateprop2}.) Every non-empty $S_i$ is either not a subset of the
% union of the later sets (which implies the existence of a $q\in S_i$ that is
% not contained in any $S_j$ with $j>i$), or it is a subset of the union of the
% later sets and non-empty, in which case all elements of $S_i$ are removed from
% the sets $S_j$ with $j>i$ (which again implies the existence of a $q\in S_i$
% that is not contained in any $S_j$ with $j>i$). This implies the claim.
%
% \ref{rem_parstateprop3}.) If $q'$ is the initial state, then the claim follows
% directly from its definition. So let $q'=\delta(q'',a)$ for some $q''$ with
% state sets $S_i''$ and some $a$.
% Let $q\in S_0$. Then, $q\in\Delta(S_0'',a)\cap T$ where $T=\{q\in Q\mid
% d(q,i)\not=\bot\text{ for every }i\}$.
%
% Now, let $q\in T=\{q\in Q\mid d(q,i)\not=\bot\text{ for every }i\}$. Then, we
% have $q\in\Delta(S_0'',a)$ by the definition of $d$. Hence,
% $q\in\Delta(S_0'',a)\cap T=S_0$.\qedhere
%
% \end{proof}

To improve readability, we say that a finite or infinite run $\rho$ satisfies~$\mathcal{O} =
\{(F_1,b_1),\ldots,(F_k,b_k)\}\subseteq 2^Q\times\nats_+$, if for
every $j$ we have that every infix of $\rho$ of length $b_j$ contains at least
one state from $F_j$. Next, we show that $d(q,j)$ counts the time since the
unique simulated run of $\aut$ ending in $q$ has visited $F_j$.

\begin{lemma}\label{lem_parrunprop} Let $q_0'q_1'q_2'\dots$ be the run of
$\autp$ on $w_0w_1w_2\dots\in\Sigma^\omega$ with $q_t' = (( S_0^t,
m_0^t), \dots, (S_{n}^t, m_{n}^t), d^t)$.
\begin{enumerate}
\item\label{lem_parrunprop1} If $q_t\in S_i^t$, then there exists a
(unique) finite run $q_0 q_1 \dots q_t$ of $\aut$ on $w_0w_1\dots w_{t-1}$ that satisfies~$\mathcal{O}$.

\item\label{lem_parrunprop2} Let $t_0<t_1$ be positions of $q_0'q_1'q_2'\dots$
and let $i$ be in the range $0\le i\le n$ such that
\begin{itemize}
\item $m_i^{t_0}=m_i^{t_1}=1$,
\item $S_i^t\not=\emptyset$ for every $t$ in the range $t_0\le t\le t_1$, and
\item $m_{i'}^t=0$ and $S_{i'}^t\not=\emptyset$ for every $t$ in the range $t_0\le
t\le t_1$ and every $i'<i$.
\end{itemize}
Then, every finite run $q_{t_0}\dots q_{t_1}$ of $\aut$ on $w_{t_0}\dots
w_{t_1-1}$ satisfying $q_t\in S_i^t$ for every $t$ in the range $t_0\le t\le
t_1$ visits a state in $F$ at least once.

\item\label{lem_parrunprop3} Let $q_0q_1q_2\dots$ be a run of $\aut$ on
$w_0w_1w_2\dots$ that satisfies~$\mathcal{O}$. Then, we have
$q_t\in S_0^t$ for every $t$.
\end{enumerate}
\end{lemma}

\begin{proof}

\ref{lem_parrunprop1}.) We show a stronger statement by induction over $t$:  if $q_t\in S_i^t$ for some $i$, then there exists a
finite run $q_0 q_1 \dots q_t$ of $\aut$ on $w_0w_1\dots w_{t-1}$
that satisfies~$\mathcal{O}$ and for every $j$ in the range $0\le j\le k$ we have $d^{t}(q_t,j)=\min\{t-t'\mid t'\le t \text{ and } q_{t'}\in F_j\}$ or (in case there is no such $q_{t'}\in F_j$) we have $d^t(q_t,j)=|q_0q_1\cdots q_t|=t+1$. Uniqueness of the run is then implied by non-confluence of $\aut$.

Due to Remark~\ref{rem_parstateprop}.\ref{rem_parstateprop1} it suffices to consider
$i=0$. The claim holds for $t=0$ by definition
of $q_0'$. Now, let $t>0$: as $q_t\in S_{0}^t$, there is a unique (due to
non-confluence) state $q_{t-1}\in S_0^{t-1}$ such that $(q_{t-1},w_{t-1},q_t)\in\Delta$.
Applying the inductive hypothesis, we obtain a run $q_0q_1\dots q_{t-1}$ of
$\aut$ on $w_0w_1\dots w_{t-2}$ that satisfies~$\mathcal{O}$
and $d^{t-1}(q_{t-1},j)=\min\{(t-1)-t'\mid t'\le t-1 \text{ and } q_{t'}\in F_j\}$ or  $d^{t-1}(q_{t-1},j)=|q_0q_1\cdots q_{t-1}|=t$.
Furthermore, Remark \ref{rem_parstateprop}.\ref{rem_parstateprop3} yields $d^t(q_t,j) < b_j$.
We consider two cases: if $q_t\in F_j$, then $q_0q_1\cdots q_t$ satisfies~$\mathcal{O}$ and
we have $d^{t}(q_t,j)=0$, by definition of $d^t$, which is equal to $\min\{t-t'\mid t'\le t \text{ and } q_{t'}\in F_j\}$.
Now, suppose $q_t\notin F_j$. Then, we have $d^{t-1}(q_{t-1},j)<b_j-1$, since we have $d^t(q_t,j)=d^{t-1}(q_{t-1},j)+1<b_j$ by the definition of $d^t$ in case $q_t\notin F_j$. We consider the two choices for the value of $d^{t-1}(q_{t-1},j)$. If
$d^{t-1}(q_{t-1},j)=\min\{(t-1)-t'\mid t'\le t-1 \text{ and } q_{t'}\in F_j\}<b_j-1$,
then the suffix of $q_0q_1\cdots q_{t-1}$ of length $b_j-1$, contains a vertex from $F_j$. Thus, also the suffix of $q_0q_1\cdots q_{t}$ of length $b_j$ contains a vertex from $F_j$ and hence
$d^{t}(q_t,j)=\min\{(t-1)-t'\mid t'\le t-1 \text{ and } q_{t'}\in F_j\}+1 = \min\{t-t'\mid t'\le t \text{ and } q_{t'}\in F_j\}$ and $q_0q_1\cdots q_t$ satisfies~$\mathcal{O}$, since the induction hypothesis applies to every infix but the last one, which has a vertex from $F_j$.
Otherwise, if $d^{t-1}(q_{t-1},j)=|q_0q_1\cdots q_{t-1}|=t<b_j-1$, then $d^t(q_t,j)=t+1=|q_0q_1\cdots q_t|$ by definition of $d^t$. Then, $q_0q_1\cdots q_t$ trivially satisfies~$\mathcal{O}$, as it has no infix of length $b_j$.

\ref{lem_parrunprop2}.) We assume $q_{t_1}\notin F$, since we are done otherwise.
We have $q_{t_0}\notin S_{i'}^{t_0}$ for every $i'>i$, due to
$m_i^{t_0}=1$, which means all states from $S_i^{t_0}$ are deleted from the sets $S_{i'}^{t_0}$
for every $i'>i$. Let $t'$ in the range $t_0<t'\le t_1$ be the first position such
that $q_{t'}\in\bigcup_{i'=i+1}^{n}S_{i'}^{t'}$. Such a position exists, as
we have $m_i^{t_1}=1$, which implies $q_{t_1}\in S_{i'}^{t_1}$ for some $i'>i$.
Since $q_{t'}\in S_{i'}^{t'}$, either
$q_{t'}\in\Delta(S_{i'}^{t'-1},w_{t'-1})$ or $q_{t'}\in
\Delta(S_0^{t'-1},w_{t'-1})\cap F$. Thus, it suffices to derive a contradiction
in the first case: $q_{t'}\in\Delta(S_{i'}^{t'-1},w_{t'-1})$ implies the
existence of a $p\in S_{i'}^{t'-1}$ such that $(p,w_{t'-1},q_{t'})\in\Delta$. We
have $p\not=q_{t'-1}$ due to the minimality of the position $t'$. But then
Lemma~\ref{lem_parrunprop}.\ref{lem_parrunprop1} yields two different runs of
$\aut$ from $q_0$ to $q_{t'}$ on $w_0\dots w_{t'-1}$, which gives the desired
contradiction to the non-confluence of $\aut$.

\ref{lem_parrunprop3}.) Again, we show a stronger statement by induction over $t$: let $q_0q_1q_2\dots$ be a run of $\aut$ on
$w_0w_1w_2\dots$ that satisfies~$\mathcal{O}$. Then, for every $t$ we have
$q_t\in S_0^t$ and for every $j$ in the range $1\le j \le k$ we have $d^t(q_t,j)=\min\{t-t'\mid t'\le t\text{ and } q_{t'}\in
F_j\}$ or (in case there is no such $q_{t'}\in F_j$) we have $d^t(q_t,j)=|q_0q_1\cdots q_t|=t+1$.
Note that this statement is only well-defined for a non-confluent automaton.

The induction start $t=0$ follows from the definition of $q_0'$. Now, let $t>0$: the induction hypothesis yields $q_{t-1}\in S_0^{t-1}$ and for every $j$ in the range $1\le j \le k$ we have $d^{t-1}(q_{t-1},j)=\min\{(t-1)-t'\mid t'\le {t-1}\text{ and } q_{t'}\in
F_j\}$ or $d^{t-1}(q_{t-1},j)=|q_0q_1\cdots q_{t-1}|=t$.
We consider two cases. If $q_t\in F_j$, then we have $q_t\in S^t_0$ and $d^t(q_t,j)=0$, which is equal to $\min\{t-t'\mid t'\le t\text{ and } q_{t'}\in F_j\}$ by definition of $S_0^t$ and $d^t$.
Otherwise, if $q_t\notin F$, then we have $d^{t-1}(q_{t-1},j)<b_j-1$, by induction hypothesis and the fact that $q_0q_1\cdots q_{t-1}$ satisfies~$\mathcal{O}$. Due to Remark \ref{rem_parstateprop}.\ref{rem_parstateprop3}, it suffices to show $d^t(q_t,j)<b_j$.
We consider the two choices for the value of $d^{t-1}(q_{t-1},j)$.
If $d^{t_1}(q_{t-1},j)=\min\{(t-1)-t'\mid t'\le {t-1}\text{ and } q_{t'}\in
F_j\}<b_j-1$, then $d^{t}(q_t,j)=\min\{(t-1)-t'\mid t'\le t-1 \text{ and } q_{t'}\in F_j\}+1 = \min\{t-t'\mid t'\le t \text{ and } q_{t'}\in F_j\}<b_j$. On the other hand, if  $d^{t-1}(q_{t-1},j)=|q_0q_1\cdots q_{t-1}|=t<b_j-1$, then $d^{t}(q_{t},j)=t+1=|q_0q_1\cdots q_{t}|<b_j$.
\qedhere
\end{proof}

We are now able to prove the correctness of Construction~\ref{cons_countpar}.
Our proof proceeds along the lines of the proof for the original construction
without counters~\cite{MS08}.

\begin{lemma}\label{lem_parcorrect} Let $\aut=(Q,\Sigma, q_0,\Delta,F)$ be a non-confluent
B\"uchi automaton, let $\{(F_1,b_1),\dots,(F_k,b_k)\}\subseteq
2^{Q}\times\nats_+$, and let $\autp$ be the
deterministic parity automaton obtained from Construction~\ref{cons_countpar}.
\begin{enumerate}
\item\label{lem_parcorrect1} $\autp$ accepts $w$ if and only if $\aut$ has an accepting run $\rho$ on $w$ such
that every $F_j$ is visited at least once in every infix of $\rho$ of length
$b_j$.
\item\label{lem_parcorrect2} $|\autp|\le 2^{(|\aut|+1)^2}\cdot \left(\Pi_{j=1}^k(b_j+1)\right)^{|\aut|}$ and
$|c(Q')|=2|\aut|+1$.
\end{enumerate}
\end{lemma}

\begin{proof}
\ref{lem_parcorrect1}.)
Let $q_0'q_1'q_2'\dots$ be an accepting run of $\autp$ on $w$, with
$q_t'=((S_0^t,m_0^t),\dots,(S_{n}^t,m_{n}^t),d^t)$. Then, there exists a
position $t_0$ and an $i$ such that $c(q_{t}')=2i$ for infinitely many $t$ and
$c(q_t')\ge 2i$ for every $t\ge t_0$. Thus, $S_{i'}^t\not=\emptyset$ for every
$t\ge t_0$ and every $i'\le i$ and $m_{i'}^t=0$ for every $t\ge t_0$ and every
$i'<i$. Since $S^{t+1}_i$ is a non-empty subset of $\Delta(S_i^{t},w_{t})$ for every $t\ge t_0$, K\"onig's Lemma yields an infinite
run $q_{t_0}q_{t_0+1}q_{t_0+2}\dots$ (not necessarily starting in an initial
state) of $\aut$ on $w_{t_0}w_{t_0+1}w_{t_{0}+2}\dots$ such that $q_t\in S_i^t$
for every $t\ge t_0$. Furthermore, there exists a finite run of $\aut$ on
$w_0\dots w_{t_0-1}$ starting in an initial state and ending in $q_{t_0}$ due to
Lemma~\ref{lem_parrunprop}.\ref{lem_parrunprop1}. These runs can be concatenated
to an infinite run $q_0q_1q_2\dots$ of $\aut$ on $w$ such that $q_t\in S_0^t$
for every $t$. Hence, $q_0q_1q_2\dots$ satisfies~$\mathcal{O}$ due to
Lemma~\ref{lem_parrunprop}.\ref{lem_parrunprop1}.
Let $t_1<t_2<t_3<\cdots$ be the positions after $t_0$ such that
$c(q_{t_s}')=2i$, i.e., $m_i^{t_s}=1$.
The run $q_0q_1q_2\dots$ is accepting due to
Lemma~\ref{lem_parrunprop}.\ref{lem_parrunprop2}, as the run visits an accepting
state in between any $t_s$ and $t_{s+1}$, of which there are infinitely many.

Now, let $q_0q_1q_2\dots$ be an accepting run of $\aut$ on $w$ that satisfies
$\mathcal{O}$ and let $q_0'q_1'q_2'\dots$ be the run of $\autp$ on $w$ with
$q_t'=((S_0^t,m_0^t),\dots,(S_{n}^t,m_{n}^t),d^t)$. We have $q_t\in S_0^t$ for
every $t$ due to Lemma~\ref{lem_parrunprop}.\ref{lem_parrunprop3}. Assume there are only
finitely many $t$ such that $m_0^t=1$. Then, there is a minimal index $i_1$ such that an
infinite suffix of $q_0q_1q_2\dots$ is tracked by $S_{i_1}$ and
$S_{i'}^t\not=\emptyset$ for every $i'\le i_1$ from some point onwards. This is due to the fact that for every $t$ with $q_t\in F$ the set $S_0\cap F$ (which contains $q_t$) is appended to the list of state sets. Furthermore, this set can be moved to the left (in case other sets are empty) only a finite number of times. Finally, if the state $q_t$ is deleted from this set, then there is a smaller set which tracks this run, for which the same reasoning applies.
Again, assume there are only
finitely many $t$ such that $m_{i_1}^t=1$. Then, there exists a minimal
index $i_2>i_1$ such that an infinite suffix of $q_0q_1q_2\dots$ is tracked by
$S_{i_2}$ and $S_{i'}^t\not=\emptyset$ for every $i'\le i_2$ from some point
onwards. This can be iterated until we have that the sets $S_{n-1}^t$ track the
suffix of $q_0q_1q_2\dots$ and all smaller sets are always non-empty. But as
$S_{n-1}^t$ is in this situation always a singleton (see 
Remark~\ref{rem_parstateprop}.\ref{rem_parstateprop2}), it gets marked every
time an accepting state is visited by $q_0q_1q_2\dots$. Hence, the run of
$\autp$ on $w$ is accepting.

\ref{lem_parcorrect2}.) Clear.\qedhere
\end{proof}

The Lemmata~\ref{lem_gnba},~\ref{lem_nba}, and~\ref{lem_parcorrect} imply the
existence of a deterministic parity automaton with the properties required in Lemma~\ref{lem_autsize}. Hence, this finishes the proof
of Theorem~\ref{thm_pltlopt}. To compute a finite-state strategy realizing the optimal value (witnessed by a valuation $\alpha$) in a $\pltlf$ game with winning condition $\phi$, one has to compute a deterministic parity automaton recognizing the $\omega$-words $w$ satisfying $\phi_\alpha$, as explained above Lemma~\ref{lem_autsize}. Dually, in a $\pltlg$ game with winning condition $\phi$, one computes a deterministic parity automaton recognizing the $\omega$-words $w$ satisfying $\neg\phi_\alpha$, which is then complemented by incrementing the priorities. This complement automaton is a memory structure for the $\pltlg$ game.

\section{Conclusion}
\label{sec_conc}
% We presented algorithms to solve infinite games with winning conditions in extensions of
% linear temporal logic that allow to specify bounds on the satisfaction of temporal
% operators.
%We adapted the alternating-color technique to graph-based games to solve
%$\prompt$ games. This result was crucial to obtain our results about $\pltl$ games:
We presented
$\twoexp$-algorithms for computing optimal strategies in a $\pltl$ game and to
 determine whether a given player wins with respect to some, infinitely  many,
or all variable valuations.
% All these problems are solved by a reduction to (in some cases
% several) $\ltl$ games. Table~\ref{tab_comp} lists the number of $\ltl$ games needed to solve in order
% to answer the problem under consideration. Note that in some cases, these are
% $\ltl$ games with blinking semantics, which are no harder to solve than classical $\ltl$
% games.
%
% \begin{table}[h]
% \centering
% \caption{Complexity of decision and optimization problems for $\pltl$}
% \begin{tabular}{l l l l l}
% 	\toprule
% 	$\phi$ & \phantom{x} & Problem & \phantom{x} & $\ltl$ games to solve \\
%
% 	\midrule
% 	$\pltl$		&& $\alpha\in\W_{\game}^0$?		&& $1$ \\
% 	$\pltl$		&& $\W_{\game}^0$ empty? 		&& $1$ \\
% 	$\pltl$ 	&& $\W_{\game}^0$ universal? 		&& $1$ \\
% 	$\pltl$ 	&& $\W_{\game}^0$ finite? 		&& $\max\{|\var(\phi)\cap\Xvar|,1\}$\\
%
% 	%\midrule
% 	$\pltlg$ 	&& $\max_{\alpha\in\W_{\game}^0}\max_{y\in\var(\phi)}\alpha(y)$
% 			&& $|var(\phi)|(\log_2(2(|\arena|f(7(|\phi|+1)))+1)+1)$ \\
%
% 	$\pltlg$ 	&& $\max_{\alpha\in\W_{\game}^0}\min_{y\in\var(\phi)}\alpha(y)$
% 			&& $\log_2(2(|\arena|f(7(|\phi|+1)))+1)+1$ \\
%
% 	%\midrule
% 	$\pltlf$ 	&& $\min_{\alpha\in\W_{\game}^0}\min_{x\in\var(\phi)}\alpha(x)$
% 			&& $|var(\phi)|(\log_2(2(|\arena|f(7(|\phi|+1)))+1)+1)$ \\
%
% 	$\pltlf$ 	&& $\min_{\alpha\in\W_{\game}^0}\max_{x\in\var(\phi)}\alpha(x)$
% 			&& $\log_2(2(|\arena|f(7(|\phi|+1)))+1)+1$ \\
%
% 	\bottomrule
% \end{tabular}
%
% \label{tab_comp}
% \end{table}
%
The decision problems for $\prompt$ and $\pltl$ (with the exception of the finiteness problem
for $\pltl$) are decidable by solving a single $\ltl$ game of the same size. Hence, adding parameterized
operators does not increase the asymptotic computational complexity of solving these games.
Furthermore, even the optimization problems for unipolar games can be solved in doubly-exponential time, so they are of the same computational complexity as solving $\ltl$ games. However, it takes an exponential number of parity games to solve to determine an optimal strategy. It is open whether this can be improved.

%But as $f$ is doubly-exponential
%Note that we refrained from analyzing the runtime of the algorithms for the optimization problems. In further research,
%one should investigate the complexity of solving the optimization
%problems. This includes giving lower bounds as well as
%faster algorithms than the ones relying on binary search presented here.
An interesting open question
concerns the tradeoff between the size of a finite-state strategy and the quality of the bounds it is winning for.
%Finally, time-optimal strategies for other winning conditions and
%approximately optimal winning strategies should be investigated.
\smallskip

\noindent\textbf{Acknowledgments.} The author wants to thank Marcin Jurdzi\'nski, Christof L\"oding, Andreas Morgenstern, and Wolfgang Thomas for helpful discussions, and Roman Rabinovich for coming up with the name \emph{blinking semantics}. Also, valuable comments on an earlier paper~\cite{Z10}
by anonymous referees are gratefully acknowledged.

\bibliographystyle{eptcs}
\bibliography{biblio}{}

\end{document}